\documentclass[12pt]{article}

\usepackage{latexsym,graphics,epsfig,amsmath,amssymb,tabularx,booktabs,color,array}

\usepackage{natbib}

\usepackage{enumerate}
\usepackage{amsthm}
\usepackage{url}

\theoremstyle{definition}
\newtheorem{example}{Example}[section]
\newtheorem{definition}{Definition}[section]
\newtheorem{remark}{Remark}[section]

\theoremstyle{plain}
\newtheorem{theorem}{Theorem}[section]
\newtheorem{lemma}[theorem]{Lemma}
\newtheorem{proposition}[theorem]{Proposition}
\newtheorem{corollary}[theorem]{Corollary}

\newenvironment{keywords}{%
\begin{changemargin}{1cm}{1cm}
\noindent{\bf Keywords:}
}
{\end{changemargin} }

\newenvironment{changemargin}[2]{%
\begin{list}{}{%
\setlength{\topsep}{0pt}%
\setlength{\leftmargin}{#1}%
\setlength{\rightmargin}{#2}%
\setlength{\listparindent}{\parindent}%
\setlength{\itemindent}{\parindent}%
\setlength{\parsep}{\parskip}%
}%
\item[]}{\end{list}}

\newcommand{\vol}{\mathrm{vol}}
\newcommand{\diag}{\textrm{diag}}
\newcommand{\var}{\textrm{var}}

\def\ci{\perp\!\!\!\perp}

\begin{document}

\title{Faithfulness and learning hypergraphs from discrete distributions}

\author{Anna Klimova \\
{\small{Institute of Science and Technology, Austria} }\\
{\small \texttt{aklimova@ist.ac.at}}\\
{}\\
\and
Caroline Uhler \\ 
{\small{Institute of Science and Technology, Austria} }\\
{\small \texttt{caroline.uhler@ist.ac.at}}\\
{}\\
\and Tam\'{a}s Rudas \\
{\small{E\"{o}tv\"{o}s Lor\'{a}nd University, Budapest, Hungary}}\\
{\small \texttt{rudas@tarki.hu}}\\
}

\date{}

\maketitle

\begin{abstract}
In this paper, we study the concepts of faithfulness and strong-faithfulness for discrete distributions. In the discrete setting, graphs are not sufficient for describing the association structure. So we consider hypergraphs instead, and introduce the concept of parametric (strong-) faithfulness with respect to a hypergraph. Assuming strong-faithfulness, we build uniformly consistent parameter estimators and corresponding procedures for a hypergraph search. The strength of association in a discrete distribution can be quantified with various measures, leading  to different concepts of strong-faithfulness. We explore these by computing lower and upper bounds for the proportions of distributions that do not satisfy strong-faithfulness.  
\end{abstract}

\begin{keywords}
contingency tables, directed acyclic graphs, hierarchical log-linear models, hypergraphs, (strong-) faithfulness
\end{keywords}

\section{Introduction}\label{intro}

A graphical model is a set of probability distributions whose association structure can be identified with a graph.  Given a  graph, the Markov property entails a set of conditional independence relations that are fulfilled by distributions in the model. Distributions in the model that obey no further conditional independence relations are called \emph{faithful to the graph}. For each undirected graphical model, as well as for each directed acyclic graph (DAG) model, there is a distribution that is faithful to the graph \citep*[cf.][]{SpirtesBook}. Moreover, the Lebesgue measure of the set of parameters corresponding to distributions that are unfaithful to a graphical model is zero; this result was proven by \cite*{SpirtesBook} for the case of multivariate normal distributions, by \cite{MeekFaith} for discrete distributions on multi-way contingency tables, and by \cite*{Pena2009} for arbitrary sample spaces and dominating measures. It is also well-known, that a DAG model may include distributions that are unfaithful to it but are not Markov to any nested DAG. This kind of unfaithfulness may occur due to path cancellation and can arise both in the discrete and in the multivariate normal settings \citep[cf.][]{ZhangSp2008, UhlerRaskutti2013}. 

In the discrete case, the non-existence of a graph to which a distribution is faithful is related to the presence of higher than first order interactions in this distribution. Graph learning algorithms \citep[cf.][]{SpirtesBook}, which do not recognize the presence of higher order interactions, may produce a graph which does not reveal the true association structure \citep[cf.][]{StudenyBook}. In order to avoid such errors, graph learning algorithms usually assume the existence of a DAG to which the distribution is faithful. Since the Lebesgue measure of the set of parameters corresponding to the distributions that are unfaithful to the underlying graph is zero, the faithfulness assumption is not considered to be restrictive in the context of graphical search. While graph search procedures assuming faithfulness are pointwise consistent, they are not uniformly consistent and thus cannot simultaneously control Type I and Type II errors with a finite sample size  \citep*{RSSWassUnifCons}. To ensure existence of a uniformly consistent learning procedure, strong-faithfulness of a distribution to the underlying DAG is needed  \citep{ZhangSpirtesLambdaFaith}. \cite*{UhlerFaithGeometry}  analyzed the Gaussian setting and showed that the strong-faithfulness assumption may, in fact, be very restrictive and the corresponding proportions of distributions which do not satisfy strong-faithfulness may become very large as the number of nodes grows. 
 
The concepts of faithfulness and strong-faithfulness were originally introduced in the causal search framework, where they are linked to identifiability of causal effects.  However, as we show in this paper using the discrete setting, these concepts are also important for identifiability of more general parameters of association. In  Section~\ref{sectionGraphFaith}, we define the concept of a model class being closed under a faithfulness relation: for each positive distribution, there exists a model in such a class to which it is faithful. By giving examples of distributions that are not faithful to any  directed or undirected graphical model, we show that these model classes are not closed under the faithfulness relation which is based on the corresponding Markov property.  Further, we introduce the concept of parametric faithfulness of a distribution to a hypergraph (instead of a graph). This concept seems more adequate for categorical data, where hypergraphs can be used to represent hierarchical log-linear models. Indeed, we show that the class of models associated with hypergraphs is closed under a parametric faithfulness relation. 

In Section~\ref{sectionStrongFaith}, we describe two major difficulties with the concept of strong-faithfulness in the discrete case. First, in contrast to role of correlations in the multivariate normal case, there is no single standard measure of the strength of association in  a joint distribution. Therefore, depending on the measure of association, different variants of strong-faithfulness may be considered. Second, the proportion of strong-faithful distributions depends on the parameterization used and can only be computed if the parameter space has finite volume. We explore the consequences of different parameterizations and measures of association for the case of the $2 \times 2$ contingency table. We define parametric strong-faithfulness with respect to a hypergraph  under a parameterization based on the log-linear interaction parameters. Assuming strong-faithfulness, we show that the maximum likelihood estimators of the interaction parameters associated with the hyperedges are uniformly consistent. As a result, we give a set of conditions, under which Type I and Type II errors can be controlled with a finite   sample size. We also discuss the uniform consistency of model selection procedures for a hypergraph search, for example, using the approaches described by \cite{EdwardsBook, EdwardsNote}.

In Section~\ref{SecProp}, we estimate the proportion of distributions that do not satisfy the parametric strong-faithfulness assumption with respect to a given hypergraph. We give an exact formulation of these proportions, under a parameterization based on conditional probabilities, for hypergraphs whose hyperedges form a decomposable set. The association structure of such distributions  may be discovered incorrectly during a hypergraph learning procedure. Finally, we define the concept of projected strong-faithfulness, which applies to distributions which do not belong to the hypergraph, and estimate the proportions of projected strong-faithful distributions for several hypergraphs for the $2 \times 2 \times 2$ contingency table. 

In Section \ref{SecConcl}, we conclude the paper with a brief discussion of our results and their implications.

\section{Graphical and parametric faithfulness} \label{sectionGraphFaith}

In this section, we first review the concept of faithfulness with respect to a graph. We then introduce parametric faithfulness with respect to a hypergraph and show that this is a more relevant concept for categorical data.

\subsection{Faithfulness with respect to a graph} 

Let $\mathcal{V}_1, \dots, \mathcal{V}_K$ be random variables taking values in $\mathcal{I} = \mathcal{I}_1 \times \dots \times \mathcal{I}_K $, a Cartesian product of finite sets. $\mathcal{I}$ describes a $K$-way contingency table and a vector $\boldsymbol i = (i_1,\dots,i_K) \in \mathcal{I}$ forms a cell.  A subset $M \subseteq \{1,\dots ,K\}$ specifies a marginal of the joint distribution of $\mathcal{V}_1, \dots, \mathcal{V}_K$, and  $M=\emptyset$ is the empty marginal. For $M=(k_1,\dots k_t)$, the set $\mathcal{I}_M = \mathcal{I}_{k_1} \times \dots \times \mathcal{I}_{k_t}$ is a \emph{marginal table}, and  the canonical projection $\boldsymbol i_M$ of the cell $\boldsymbol i$ onto the set $\mathcal{I}_M$ is a \emph{marginal cell}. We parameterize the population distribution by cell probabilities $\boldsymbol p =(p_{\boldsymbol i})_{\boldsymbol i\in \mathcal{I}}$, where $p_{\boldsymbol i} \in (0,1)$ and $\sum_{\boldsymbol i \in \mathcal{I} }p_{\boldsymbol i} = 1$, and denote by $\mathcal{P}$ the set  of all distributions on $\mathcal{I}$. A subset of $\mathcal{P}$ is called a \emph{model}. For simplicity of exposition, we assume that $\mathcal{V}_1, \dots, \mathcal{V}_K$ are binary, $\mathcal{I}$ is treated as a sequence of cells ordered lexicographically, and 
a distribution $P\in \mathcal{P}$ is addressed  by its parameter, $\boldsymbol p$. 

A \emph{graphical model} is a set of probability distributions, whose association structure can be identified with a graph with vertices $V = \{1, \dots, K\}$, where each vertex $i$ is associated with a random variable $\mathcal{V}_i$. In the following, we will identify each vertex with its associated random variable. The absence of an edge between two vertices means that the corresponding random variables satisfy some (conditional) independence relation. A detailed description of graphical models for discrete as well as for multivariate normal distributions can be found in \cite{EdwardsBook}, among others. In the sequel, we only consider undirected graphical models and DAG models.

A graphical model identified with an undirected graph (also called a \emph{graphical log-linear model} in the discrete setting) is a set of probability distributions on $V$ that satisfy the \emph{local undirected Markov property}: Every node is conditionally independent of its non-neighbors given its neighbors. In the discrete case, such models are a sublcass of hierarchical log-linear models. A graphical model identified with a directed acyclic graph, a DAG model, is a set of probability distributions on $V$ that satisfy the \emph{directed Markov property}: Every node is conditionally independent of its non-descendants given its parents. A distribution that satisfies the Markov property with respect to a graph is called \emph{Markov} to it. 

A distribution which is Markov to a graph, is said to be \emph{faithful} to it if all conditional independencies in this distribution can be derived from the graph. The faithfulness relation can be thought of as a decision rule that classifies a distribution $\boldsymbol p$ in a model $\mathcal{M}$  as faithful or unfaithful to it:
$$\mathbb{F}(\boldsymbol p, \mathcal{M}) = \left\{\begin{array}{ll} 1 & \mbox{ if } \boldsymbol p \mbox{ is faithful to } \,  \mathcal{M},\\
0 & \mbox{ otherwise}.\\
 \end{array}\right.$$

\begin{definition} 
A class $\frak{C}$ of models on $\mathcal{P}$, where $\frak{C}$ is partially ordered with respect to inclusion, is said to be \emph{closed} under the faithfulness relation indicated by $\mathbb{F}$, if for every non-empty $\mathcal{M} \in \frak{C}$ and for every $\boldsymbol p \in \mathcal{M}$ such that $\mathbb{F}(\boldsymbol p , \mathcal{M}) = 0$, there exists an $\mathcal{M}' \in \frak{C}$ with $\mathcal{M}' \subset \mathcal{M}$ and  $\mathbb{F}(\boldsymbol p , \mathcal{M}') = 1$.
\end{definition}

This definition implies that a class $\frak{C}$ is closed under the faithfulness relation indicated by $\mathbb{F}$ if and only if for every $\boldsymbol p \in \mathcal{P}$ there exists an $\mathcal{M} \in \frak{C}$, such that $\mathbb{F}(\boldsymbol p, \mathcal{M}) = 1$. Graphical log-linear models and DAG models are specified by a list of conditional independence relations which, in turn, comprise other conditional independencies. Thus, these model classes have a natural partial order implied by  the conditional independence relation. We now show that these classes are not closed under the corresponding faithfulness relations. 

\begin{proposition}
\label{prop_undirected}
The class of graphical log-linear models is not closed under the faithfulness relation defined by the local undirected Markov property.
\end{proposition}

The following example is given as a proof. 

\begin{example}\label{FourVarUnfaith}
Let $V=\{A, B, C, D\}$ and consider the log-linear model $[ABC][ABD]$ \citep[cf.][]{Agresti2002}. This is the model of conditional independence of $C$ and $D$ given $A$ and $B$. All distributions in this model are Markov to the graph in Figure \ref{CiDgivABgraph}. Consider the distribution parameterized by  
\begin{eqnarray*}
\boldsymbol p &=& (0.022, 0.062, 0.063, 0.103, 0.103, 0.063, 0.062, 0.022, \\ 
&&0.103, 0.063, 0.062, 0.022,0.022,0.062,0.063,0.103)',
\end{eqnarray*}
where the cell probabilities are ordered lexicographically. In this distribution, the conditional odds ratios ($\mathcal{COR}$) of $C$ and $D$ given the levels of $A$ and $B$ are equal to $1$:
\begin{equation*}
\mathcal{COR}(CD\mid A=i, B=j) = \frac{p_{ij00}p_{ij11}}{p_{ij01}p_{ij10}} = 1, \,\, \mbox{for all } i, j \in \{0, 1\}.
\end{equation*}
Hence, the distribution is in the model. The $(A,B)$-marginal of this distribution is uniform:
$$
\begin{array}{c|cc}
 & {B}=0 & {B}=1 \\
\hline
{A} = 0& 1/4 & 1/4\\
{A} = 1& 1/4 & 1/4\\
\end{array},
$$
and thus $A \ci B$. So the distribution is unfaithful to the graph in Figure \ref{CiDgivABgraph}. In addition, since the conditional odds ratios of $A$, $B$ and $C$ given $D$ and of $A$, $B$, and $D$ given $C$ are not equal to $1$:
\begin{eqnarray*}
\mathcal{COR}(ABC\mid D=0) &=& \frac{p_{0000}p_{1100}p_{0110}p_{1010}}{p_{0100}p_{1000}p_{0010}p_{1110}} \approx  0.04418483,\\
\mathcal{COR}(ABC\mid D=1) &=& \frac{p_{0001}p_{1101}p_{0111}p_{1011}}{p_{0101}p_{1001}p_{0011}p_{1111}} \approx  0.04418483,\\
\mathcal{COR}(ABD\mid C=0) &=& \frac{p_{0000}p_{1100}p_{0101}p_{1001}}{p_{0100}p_{1000}p_{0001}p_{1101}} \approx 0.04710518,\\
\mathcal{COR}(ABD\mid C=1) &=& \frac{p_{0010}p_{1110}p_{0111}p_{1011}}{p_{0110}p_{1010}p_{0011}p_{1111}} \approx 0.04710518,
\end{eqnarray*}
the distribution cannot be Markov to any nested undirected graph. 
\qed
\end{example}

The situation described in Example \ref{FourVarUnfaith} is distinctive to discrete distributions. In the Gaussian setting, marginal independence of more than two variables implies their joint independence. Thus, a multivariate normal distribution whose components are pairwise independent is Markov and faithful to a graph with no edges. But in the discrete case, a joint distribution of pairwise independent variables may have a non-trivial structure of higher than first order interactions.  Next, we prove that also the class of DAG models is not closed with respect to the faithfulness relation.

\begin{proposition}
The class of  DAG models is not closed under the faithfulness relation defined by the directed Markov property.
\end{proposition}

As a proof, two examples are given. The second example only pertains to the discrete case.

\begin{figure}[t]
\begin{minipage}[b]{0.45\linewidth}
\centering
\includegraphics[scale=0.7]{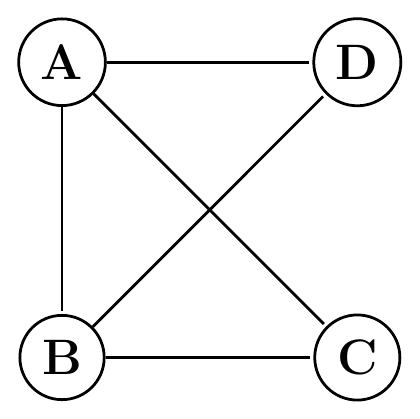}
\caption{$C \ci D \mid A,B$.}
\label{CiDgivABgraph}
\end{minipage}
\hspace{1.5cm}
\begin{minipage}[b]{0.45\linewidth}
\centering
\includegraphics[scale=0.7]{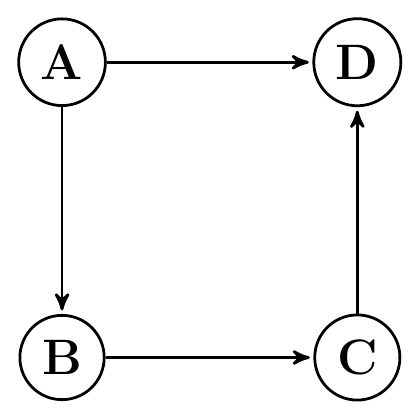}
\caption{$A \ci C \mid B, \quad B \ci D \mid A,C$.}
\label{DAGcycle}
\end{minipage}
\end{figure}

\begin{example}\label{normalDAGunfaith}
Let $V=\{A, B, C, D\}$ and consider the model specified by two conditional independence relations: $A \ci C \mid B$ and $B \ci D \mid A,C$.  Any distribution in this model is Markov to the DAG in Figure  \ref{DAGcycle}. For example, the distribution parameterized by 
\begin{eqnarray*}
\boldsymbol p &=& (0.006, 0.006, 0.0288, 0.0192, 0.06, 0.06, 0.072, 0.048, 0.0056, \\
&&0.0504, 0.187148, 0.0368516, 0.021, 0.189, 0.175452, 0.0345484)',
\end{eqnarray*}
is in the model. However, this distribution also satisfies the additional independence relation $A \ci D$.  This independence relation is not reflected in the graph. Thus the distribution is unfaithful to the graph in Figure \ref{DAGcycle}. 

Next, we show that there is no DAG that fulfills all three (conditional) independence relations $A \ci D$, $A \ci C\mid B$, $B \ci D\mid A,C$. If such a DAG existed, then its skeleton would have three edges: $AB$, $BC$, $CD$. In order to satisfy faithfulness, $A \ci D$ requires that $A\to B\leftarrow C$ or $B\to C\leftarrow D$. However, $A \ci C\mid B$ is unfaithful to $A\to B\leftarrow C$ and $B \ci D\mid A,C$ is unfaithful to $B\to C\leftarrow D$.
 \qed
\end{example}

\begin{remark} 
\label{rem_Gaussian}
One can also construct an instance of Example \ref{normalDAGunfaith} using multivariate normal distributions by choosing the partial correlations in such a way that the causal effect associated with the edge $A\to D$ cancels with the causal effect associated with the path $A\to B\to C\to D$ (see Figure \ref{DAGcycle}). This shows that also Gaussian  DAG models are not closed under the faithfulness relation defined by the directed Markov property.
\end{remark}

The next example illustrates a situation that occurs only in the discrete case. To construct this example, we will use the fact that, in contrary to the Gaussian case, a discrete distribution with pairwise independent random variables can have non-vanishing interactions of higher than the first order.

\begin{example}\label{ExampleIntro}
Let $V=\{A, B, C\}$.  Consider the distribution  parameterized by 
\begin{equation}\label{distr18}
\boldsymbol p = (1/8-\delta,1/8+\delta, 1/8+\delta,1/8-\delta, 1/8+\delta,1/8-\delta, 1/8-\delta,1/8+\delta)',
\end{equation}
where  $\delta \in (-1/8, 1/8)$. Its marginals are uniform resulting in pairwise independence: $A \ci B$, $A \ci C$, and $B \ci C$. The second order odds ratio of this distribution,
\begin{equation*}
\frac{p_{000}p_{011}p_{101}p_{110}}{p_{001}p_{010}p_{100}p_{111}} = \left(\frac{1/8-\delta}{1/8+\delta}\right)^{4},
\end{equation*} 
does not vanish, implying that $A$, $B$, and $C$ are not jointly independent. The  distribution belongs to the graphical log-linear model that can be identified with the graph shown in row 1 of Table \ref{allgraphsABC}. Further, since each pairwise independence holds, the distribution is Markov to the DAGs shown in rows 2, 3, and 4 of Table \ref{allgraphsABC}. However, the distribution is not faithful to these DAGs and it is not Markov to any of the nested DAGs (rows 5, 6, 7 and 8). \qed 
\end{example}

The association structure of a distribution that is unfaithful to every model in a given class can be considered within a larger model class. We have described examples of discrete distributions for which there is no undirected graphical model or DAG model to which they are faithful. Graphical models (directed and undirected) for discrete distributions are a subclass of hierarchical marginal log-linear models \citep*{RudasBergsma, RudasBN2006} and can be considered within this larger class. We revisit Example \ref{normalDAGunfaith} to motivate the introduction of \emph{parametric faithfulness}, a generalization of the concept of faithfulness that can be applied to the class of hierarchical marginal log-linear models. In the following, we show that under this natural generalization of faithfulness, we can find a model in the class of hierarchical marginal log-linear models to which the distribution described in Example \ref{normalDAGunfaith} is faithful.

\vspace{5mm}

\textbf{Example \ref{normalDAGunfaith}} (revisited):
A marginal log-linear parameterization  \citep{RudasBergsma} for the DAG in Figure~\ref{DAGcycle} can be derived from the set of marginals $$\mathcal{M} = \{(A, D), (A, B, C), (A, B, C, D)\}.$$ The corresponding parameters are:
\begin{eqnarray}\label{MargEx}
\lambda_{\emptyset}^{AD}, \, \,  \lambda_{A*}^{AD}, \, \,  \lambda_{*D}^{AD}, \, \,  \lambda_{AD}^{AD},  \, \,  \lambda_{*B*}^{ABC},  \, \,  \lambda_{**C}^{ABC},  \, \,  \lambda_{*BC}^{ABC},  \, \, \lambda_{AB*}^{ABC}, \, \, \lambda_{A*C}^{ABC},  \, \, \lambda_{ABC}^{ABC},  \nonumber \\
\\
\lambda_{*B*D}^{ABCD},  \, \, \lambda_{**CD}^{ABCD},  \, \,  \lambda_{AB*D}^{ABCD},  \, \,  \lambda_{A*CD}^{ABCD},  \,  \lambda_{*BCD}^{ABCD},  \,  \, \lambda_{ABCD}^{ABCD}. \nonumber
\end{eqnarray}
The conditional independencies $A \ci C \mid B$ and $B \ci D \mid A,C$ are obtained by taking
\begin{eqnarray}\label{MargEx2}
\lambda_{A*C}^{ABC} = 0, \,\, \lambda_{ABC}^{ABC} = 0, \,\, \lambda_{*B*D}^{ABCD} = 0, \,\,  \lambda_{AB*D}^{ABCD} = 0, \,\, \lambda_{*BCD}^{ABCD} = 0, \,\, \lambda_{ABCD}^{ABCD} = 0.
\end{eqnarray}
Any distribution that is Markov to the DAG in Figure  \ref{DAGcycle} can be parameterized by the remaining marginal log-linear parameters. 
The faithfulness relation in the class of marginal log-linear models can be defined as a relationship between the parameters of a distribution and the parameters of a model that contains the distribution. A distribution which also satisfies the marginal independence $A \ci D$, has $\lambda_{AD}^{AD} = 0$ and thus belongs to a nested marginal log-linear model, to which it is faithful in the parametric sense. \qed

\vspace{5mm}

This example motivates taking a parametric approach (instead of a graphical approach) to faithfulness. In the next section we introduce the concept of parametric faithfulness for discrete distributions more formally.

\subsection{Parametric faithfulness}
\label{sec_par_faith}

Let $\mathcal{P}$ denote the full exponential family of distributions. We choose a mixed parameterization $(\boldsymbol \mu, \boldsymbol \nu)$ of this family, where $\boldsymbol \mu$ denotes the vector of mean value parameters and $\boldsymbol \nu$ the vector of canonical parameters \citep[cf.][]{Barndorff1978}. Let $\frak{C}$ be a class of partially ordered exponential families  that are obtained by setting some of the components of $\boldsymbol \mu$ and/or some of the components of $\boldsymbol \nu$ to zero, and let $\mathcal{M} \in \frak{C}$. Assume that $\mathcal{M}$ is parameterized by $(\boldsymbol \mu_{\mathcal{M}}, \boldsymbol \nu_{\mathcal{M}})$, where $\boldsymbol \mu_{\mathcal{M}} \subseteq \boldsymbol \mu$, $\boldsymbol \nu_{\mathcal{M}} \subseteq \boldsymbol \nu$,  $\boldsymbol \mu \setminus \boldsymbol \mu_{\mathcal{M}} = \boldsymbol 0$, and $\boldsymbol \nu \setminus \boldsymbol \nu_{\mathcal{M}} = \boldsymbol 0$. We define faithfulness as a relationship between the parameters of a distribution and the parameters of a model containing the distribution under consideration.

\begin{definition} \label{FaithNormalDef}
A distribution $\boldsymbol p \in \mathcal{M}$ parameterized by $(\boldsymbol \mu_{\mathcal{M}}(\boldsymbol p), \boldsymbol \nu_{\mathcal{M}}(\boldsymbol p))$, satisfies the \emph{parametric faithfulness relation with respect to} $\mathcal{M}$ if none of the components of $\boldsymbol \mu_{\mathcal{M}}(\boldsymbol p)$ or  $\boldsymbol \nu_{\mathcal{M}}(\boldsymbol p)$ vanish. 
\end{definition}

The class of discrete exponential families, where the canonical parameters are the interactions of the variables in $V$ of order up to $K-1$, corresponds to the class of hierarchical log-linear models on $V$. More precisely, let $\mathcal{M}=\{M_1, \dots,M_T\}$ be a set of incomparable subsets of $V$. Then the hierarchical log-linear model generated by $\mathcal{M}$ is the set of distributions in $\mathcal{P}$ that satisfy
\begin{equation}\label{LLMdef}
\mbox{log } p_{\boldsymbol i} = \sum_{M \subseteq V: M \subseteq M_j \in \mathcal{M}} \gamma_M(\boldsymbol i_M), 
\end{equation} 
where $\gamma_{M'}(\boldsymbol i_{M'}) = 0$ implies  $\gamma_{M''} (\boldsymbol i_{M''}) = 0$ for any $M''\supseteq M'$, and $\gamma_{M}$ are called the \emph{interaction parameters} (interactions for short). Their identifiability  is assumed in the sequel. 

The set $\mathcal{M}$ partitions the power set of  $V$ into a descending class, consisting of subsets of $M_1, \dots, M_T$, and a complementary ascending class. The partition induces a mixed parameterization of $\mathcal{P}$ with the canonical parameters equal to the conditional odds ratios (or their logarithms) of the subsets in the ascending class, given the remaining variables, and the mean value parameters equal to the marginal distributions of the subsets in the descending class.  Under this parameterization, the canonical parameters of the distributions in the model generated by $\mathcal{M}$ are equal to $1$ (or zero) and the distributions are parameterized by the mean value parameters \citep{RudasSAGE}. The structure of the highest order interactions of the distributions in the  hierarchical log-linear model generated by  $\mathcal{M}=\{M_1, \dots,M_T\}$ is described next. In the sequel, $\bar{M}_t= V \setminus M_t$.

\begin{lemma}\label{interactions}
There exists a parameterization of $\mathcal{P}$ under which, for every $t = 1, \dots, T$, the interaction parameter $\gamma_t$ is equal to the logarithm of the conditional odds ratio of $M_t$ given  $\bar{M}_t = \boldsymbol i_{\bar{M}_t}$, and is invariant of the choice of $\boldsymbol i_{\bar{M}_t}$.
\end{lemma}

\begin{proof}
There exists a marginal log-linear parameterization of $\mathcal{P}$ under which for every $t = 1, \dots, T$, the interaction parameter $\gamma_{t}$,  corresponding to the generating marginal $M_t$, is the average log conditional odds ratio of $M_t$ conditioned on and averaged over $\bar{M}_t$ \citep[cf.][]{RudasBN2006}:  
$$\gamma_{M_t} = \frac{1}{|\mathcal{I}_{\bar{M}_t}|}\sum_{\boldsymbol i_{\bar{M}_t}}\mbox{log } \mathcal{COR}(M_t\mid \bar{M}_t = \boldsymbol i_{\bar{M}_t}).$$  Since $M_{t}$ is a  maximal interaction, $\mathcal{COR}(M'\mid \bar{M}' = \boldsymbol i_{\bar{M}'}) = 1,$ for any $M' \supsetneq M_t$.

Further, it can be shown by induction on the elements of the ascending class of $M_1, \dots, M_T$, that 
$$\mathcal{COR}(M'\mid\bar{M}' = \boldsymbol i_{\bar{M}'}) = \frac{ \mathcal{COR}(M_t \mid (M'\setminus M_t)\cup\bar{M}' =(\boldsymbol i_{M'\setminus M_t}, \boldsymbol i_{\bar{M}'}))}{\mathcal{COR}(M_t \mid (M'\setminus M_t)\cup\bar{M}' =(\boldsymbol j_{M'\setminus M_t}, \boldsymbol i_{\bar{M}'}))} = \frac{\mathcal{COR}(M_t \mid \bar{M}_t = \boldsymbol i_{\bar{M}_t})}{\mathcal{COR}(M_t \mid \bar{M}_t = \boldsymbol j_{\bar{M}_t})},$$
and thus,
$$\mbox{log } \mathcal{COR}(M_t\mid \bar{M}_t = \boldsymbol i_{\bar{M}_t}) = \mbox{log } \mathcal{COR}(M_t\mid \bar{M}_t = \boldsymbol j_{\bar{M}_t}),$$
for any  $\boldsymbol i_{\bar{M}_t}$ and $\boldsymbol j_{\bar{M}_t}$. Hence,
$$\gamma_{M_t} =  \mbox{log } \mathcal{COR}(M_t\mid \bar{M}_t = \boldsymbol i_{\bar{M}_t}),$$
for any $\boldsymbol i_{\bar{M}_t}$.
\end{proof}

The association structure of a discrete distribution in a hierarchical log-linear model generated by $\mathcal{M}$ can be described  with a hypergraph, $\mathcal{H}= \mathcal{H}(\mathcal{M})$ with vertices $V=\{1, \dots, K\}$ and hyperedges equal to the generating marginals, or, equivalently, to the maximum non-vanishing interactions in $\mathcal{M}$. Faithfulness to a hypergraph is naturally defined as follows:

\begin{definition}\label{HypergrFaithDef}
A distribution is \emph{faithful to a hypergraph} $\mathcal{H}$ if the non-vanishing maximal interactions of this distribution coincide with hyperedges of $\mathcal{H}$.
\end{definition}

This definition implies that a distribution in the log-linear model generated by $\mathcal{M}=\{M_1,\dots M_T\}$ is faithful to the hypergraph with hyperedges $M_1,\dots M_T$ if, for all $t \in \{1,\dots, T\}$, none of the conditional odds ratios of $M_t$ given the variables in $\bar{M}_t = V \setminus M_t$ is equal to $1$.  In the following result, we show that the class of hypergraphs is closed under the parametric faithfulness relation.

\begin{theorem} \label{pIII} 
The class of hypergraphs in $\mathcal{P}$ is closed under the faithfulness relation specified by Definition \ref{HypergrFaithDef}. 
\end{theorem}

\begin{proof}
Let $P \in \mathcal{P}$. In the following, we show that there exists a hypergraph to which $P$ is faithful.  First, derive the ascending class, $\mathcal{A}$, of subsets of $V$ such that the log conditional odds ratios of the elements of $\mathcal{A}$ given the remaining variables vanish on $P$. Next, find the maximal (with respect to inclusion) elements,  $M_1, \dots, M_T$, of the complement of $\mathcal{A}$. Then, by construction, $P$ is faithful to the hypergraph with hyperedges $M_1, \dots,M_T$.  
\end{proof}

\begin{remark} 
This paper is solely concerned with discrete distributions. However, it is worth pointing out that Definition \ref{FaithNormalDef} makes sense for exponential families in general. In particular, multivariate normal distributions can be described using an exponential family whose canonical parameters correspond to pairwise interactions between the random variables in $V$. We mentioned in Remark \ref{rem_Gaussian} that there are examples of distributions in Gaussian DAG models that are not Markov to any nested DAG, and hence the class of Gaussian DAG models is not closed under the faithfulness relation. However, the class of  multivariate normal exponential families  is closed under the parametric faithfulness relation. This is the case since setting an additional canonical parameter to zero leads to a nested exponential family.
\end{remark}

\section{Parametric Strong-Faithfulness}\label{sectionStrongFaith}

In order to test statistical hypotheses when working with data, a stronger version of faithfulness is needed. In this section, we generalize the notion of parametric faithfulness to parametric strong-faithfulness and discuss difficulties arising with this concept in the discrete setting.

\begin{table}[b!]
\centering
\caption{Selected parameterizations, measures of association, and strong-faithfulness conditions for the $2 \times 2$ contingency table.}
\vspace{6mm}

\begin{tabular}{l|p{20mm}|p{26mm}|p{59mm}}
\hline 
&  & & \\
Parametrization & Parameter space& Variation independence&  Association function;  $\lambda$-strong-faithfulness condition  \\ 
\hline 
& &  &\\ 
Cell probabilities: &  & &  \\ 
$p_{00}, p_{01},$  & \multicolumn{1}{c|}{Simplex} & \multicolumn{1}{c|}{No} & $\phi_1 = \left|\mbox{log} \left(\frac{p_{00}p_{11}}{p_{01}p_{10}} \right)\right| > \lambda$  \\ [5pt]
$p_{10}, p_{11}$& \multicolumn{1}{c|}{$\Delta_3$} & \\ 
 & & & $\phi_2 = \left|\frac{p_{00}p_{11} - p_{01}p_{10}}{p_{00}p_{11} + p_{01}p_{10}}\right| > \lambda$  \\
&  & &\\ 
\hline 
& &  &\\ 
Conditional probabilities: & &  &\\  [5pt]
$\theta_1 = \mathbb{P}(A=0)$,  & \multicolumn{1}{c|}{$(0,1)^3$} & \multicolumn{1}{c|}{Yes} & $\phi_3 = |\theta_2 - \theta_3| > \lambda$  \\
$\theta_2 = \mathbb{P}(B=0\mid A=0)$, & & \\
 $\theta_3 = \mathbb{P}(B=0\mid A=1)$ & &   \\ 
\hline
\end{tabular}
\label{StrFtwoway}
\end{table}

\subsection{Strong-faithfulness in the discrete setting}
\label{subsec_strong_faith}

A distribution in a model is faithful to it if the model fully describes the conditional independence structure in this distribution. It is further called \emph{strong-faithful} if the conditional dependencies present in the distribution are strong enough. The concept of strong-faithfulness, originally defined by  \cite{ZhangSpirtesLambdaFaith}, is usually applied to multivariate normal distributions: For a given $\lambda > 0$, a multivariate normal distribution in a DAG model is $\lambda$\emph{-strong-faithful} with respect to this DAG if all non-zero partial correlations are bounded away from zero by $\lambda$. A formal definition of strong-faithfulness in the discrete case has not been proposed, although some analogies  were used. For example, \cite*{Zuk} made use of the assumption that the conditional probabilities in a Bayesian network are bounded between $\lambda$ and $1 - \lambda$. This can be seen as a form of strong-faithfulness. 

In the discrete setting, one problem is that many variants of strong-faithfulness relations can be considered. Whether a  distribution is $\lambda$-strong-faithful to a model, depends on the choice of parameterization and the measure of association. This is illustrated in the following example for two binary random variables.

\begin{example}\label{22param}
Let $V = \{A, B\}$ and consider the saturated model $[AB]$, which allows for interaction between $A$ and $B$. A distribution in which this interaction vanishes is unfaithful to $[AB]$ and belongs to the model of independence, $A \ci B$. A distribution in which the association between $A$ and $B$ is strong enough is called strong-faithful to $[AB]$. While in the multivariate normal setting the partial correlations are a standard measure of association, in the discrete setting there are many viable choices of association measures, see \cite{GoodmanKruskal74}. Table \ref{StrFtwoway} illustrates different possible definitions of strong-faithfulness based on three different measures of association, the log odds ratio, $\phi_1$, Yule's coefficient of association, $\phi_2$, and the absolute difference between the conditional probabilities, $\phi_3$.  In all three cases, the parameter space has finite volume. So it is possible to estimate the proportion (relative volume) of distributions that do not satisfy the $\lambda$-strong-faithfulness relation with respect to $[AB]$. Figure \ref{2by2Prop} shows that this proportion varies considerably depending on the chosen parameterization and association measure.
\qed
\end{example}

\begin{figure}[b!]
\centering
\includegraphics[scale=0.44]{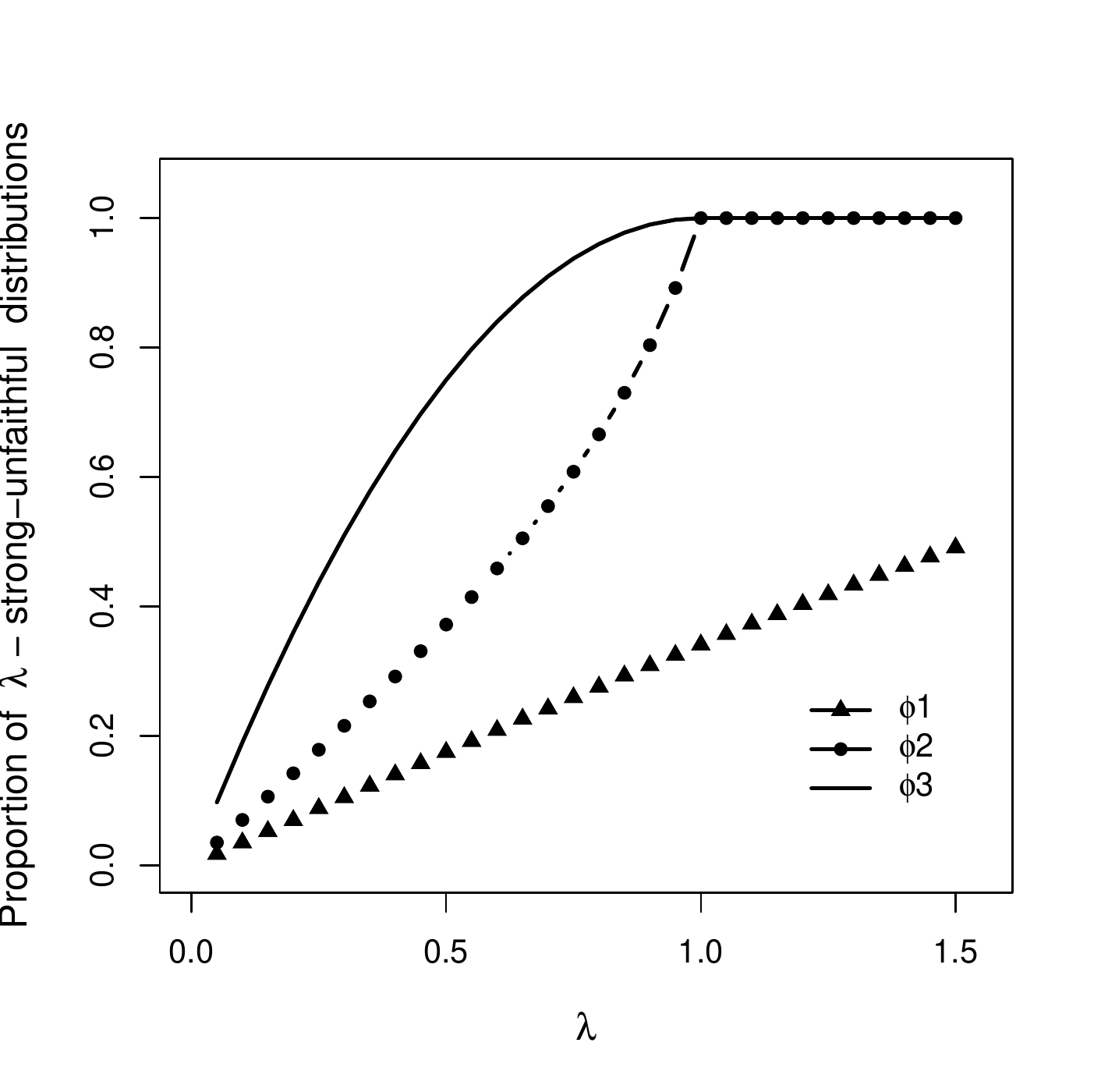}
\caption{The proportions of distributions that are not $\lambda$-strong-faithful to the model $[AB]$ with respect to different association measures, see Example \ref{22param}.}
\label{2by2Prop}
\end{figure}

The proportion of distributions in a model that do not satisfy the strong-faithfulness relation with respect to this model is of importance for model selection procedures, which are often based on the strong-faithfulness assumption. Lemma \ref{interactions} justifies the use of $\phi_1$ to define strong-faithfulness in the discrete case. In the following, we propose the concept of strong-faithfulness to a hypergraph and, assuming strong-faithfulness, prove existence of uniformly consistent estimators of the hypergraph parameters.  

\subsection{Strong-faithfulness with respect to a hypergraph}

Let $\mathcal{H}$ be the hypergraph generated by a set of marginals $\mathcal{M}= \{M_1, \dots, M_T\}$. For $\boldsymbol p \in \mathcal{H}$ let $\boldsymbol \gamma(\boldsymbol p) = (\gamma_1(\boldsymbol p), \dots, \gamma_T(\boldsymbol p))$ denote the set of  interaction parameters of $\boldsymbol p$ corresponding to the hyperedges of  $\mathcal{H}$. 

\begin{definition}
For $\lambda >0$, a distribution $\boldsymbol p \in \mathcal{H}$ is $\lambda$\emph{-strong-faithful}  to $\mathcal{H}$ if 
\begin{equation}\label{tube}
\operatorname{min} \{|\gamma_1(\boldsymbol p)|, \dots, |\gamma_T(\boldsymbol p)|\} > \lambda.
\end{equation}
\end{definition}

As described in Section \ref{subsec_strong_faith}, one can, in principle, use different measures of association to define strong-faithfulness. The advantage of the definition given here is that it generalizes the original definition of strong-faithfulness given by \cite{ZhangSpirtesLambdaFaith}. For a hypergraph generated by two-way marginals the interactions  $\boldsymbol \gamma(\boldsymbol p)$ are analogous to partial correlations of a multivariate normal distribution \citep[cf.][]{WermuthAnalogies}.  Therefore, the definition of strong-faithfulness to a hypergraph proposed here is consistent with the original definition of strong-faithfulness of a multivariate normal distribution with respect to a DAG given by \cite{ZhangSpirtesLambdaFaith}. In addition, as we will show in Section \ref{sec_hypergraph_search}, strong-faithfulness with respect to a hypergraph allows to build uniformly consistent algorithms for learning hypergraphs. 

In the following example, we illustrate the concept of strong-faithfulness with respect to a hypergraph for distributions on the $2 \times 2 \times 2$ contingency table.

\begin{example}\label{222marginal} 
Let $V = \{A, B, C\}$. A distribution of $V$ can be parameterized by 
\begin{equation*}
\mbox{log } \boldsymbol p = \mathbf{M} \boldsymbol \gamma,
\end{equation*}
where
\begin{equation*}
\mathbf{M}=\left(\begin{array}{cccccccc}
1 & 0 & 0 & 0 & 0 & 0 & 0 & 0 \\ 
 1 & 0 & 0 & 1 & 0 & 0 & 0 & 0 \\ 
 1 & 0 & 1 & 0 & 0 & 0 & 0 & 0 \\ 
 1 & 0 & 1 & 1 & 0 & 0 & 1 & 0 \\ 
 1 & 1 & 0 & 0 & 0 & 0 & 0 & 0 \\ 
 1 & 1 & 0 & 1 & 0 & 1 & 0 & 0 \\ 
 1 & 1 & 1 & 0 & 1 & 0 & 0 & 0 \\ 
 1 & 1 & 1 & 1 & 1 & 1 & 1 & 1 
\end{array}\right),
\end{equation*}
and 
$$\boldsymbol \gamma = (\gamma^{\emptyset}, \gamma^{A}_{1}, \gamma^{B}_{1}, \gamma^{C}_{1} , \gamma^{AB}_{11}, \gamma^{AC}_{11}, \gamma^{BC}_{11}, \gamma^{ABC}_{111})$$
are the interaction parameters corresponding to the marginal distributions indicated in the superscript. The matrix $\mathbf{M}$ is of full rank, and it can easily be shown that
\begin{align*}
&\gamma^{\emptyset} = \mbox{log } p_{000}, \hspace{14mm}
\gamma^{A}_{1}  = \mbox{log } \frac{p_{100}}{p_{000}} ,\\ 
&\gamma^{B}_{1} = \mbox{log } \frac{p_{010}}{p_{000}}, \hspace{13mm}
\gamma^{C}_{1} = \mbox{log }  \frac{p_{001}}{p_{000}},\\
&\gamma^{AB}_{11} = \mbox{log } \frac{p_{000}p_{110}}{p_{010}p_{100}}, \quad 
\gamma^{AC}_{11} = \mbox{log }  \frac{p_{000}p_{101}}{p_{001}p_{100}}, \\
&\gamma^{BC}_{11} = \mbox{log } \frac{p_{000}p_{011}}{p_{001}p_{010}}, \quad \gamma^{ABC}_{111} = 
\mbox{log } \frac{p_{001}p_{010}p_{100}p_{111}}{p_{000}p_{011}p_{101}p_{110}}.
\end{align*}
\vspace{0cm}

\noindent In the following table we give the $\lambda$-strong-faithfulness conditions for several hypergraph models:

\begin{equation*}
\begin{tabular}{l| l}
{Hyperedges} & \multicolumn{1}{c}{Strong-faithfulness constraints} \\
\hline  
& \\
$\{ABC\}$ &  $|\gamma^{ABC}_{111}| > \lambda$  \\ [5pt]
$\{AB\}, \{AC\}, \{BC\}$ & $\operatorname{min} \{|\gamma^{AB}_{11}|, |\gamma^{AC}_{11}|, |\gamma^{BC}_{11}|\} > \lambda$ \\ [5pt]
$\{AC\}, \{BC\}$ & $\operatorname{min} \{|\gamma^{AC}_{11}|, |\gamma^{BC}_{11}|\} > \lambda$ \\ [5pt]
$\{A\}, \{BC\}$ & $\operatorname{min} \{|\gamma^{A}_{1}|, |\gamma^{BC}_{11}|\} > \lambda$ \\ [5pt]
$\{A\}, \{B\}, \{C\}$ & $\operatorname{min} \{|\gamma^{A}_{1}|, |\gamma^{B}_{1}|, |\gamma^{C}_{1}| \} > \lambda$ \\ [5pt]
\end{tabular}
\end{equation*}
\end{example}

\subsection{Hypergraph search}
\label{sec_hypergraph_search}

In this section, we discuss how to construct hypothesis tests, when the association measure is based on the interaction parameters $\boldsymbol{\gamma}(\boldsymbol p)$, and how to perform a hypergraph search based on these hypothesis tests.

Let $\mathcal{H}$ be a hypergraph generated by the marginals $M_1,\dots , M_T$ and let $\gamma_1,\dots ,\gamma_T$ be the corresponding interaction parameters. We denote by $\mathcal{H}_{\lambda}$ the set of distributions that are $\lambda$-strong-faithful to the hypergraph $\mathcal{H}$, i.e.,
$$\mathcal{H}_{\lambda} = \{\boldsymbol p \in \mathcal{H}: \,\, \operatorname{min}\{|\gamma_1(\boldsymbol p)|, \dots, |\gamma_T(\boldsymbol p)|\} > \lambda\},$$
and define 
$$\mathcal{H}_{\lambda, \delta} = \mathcal{H}_{\lambda} \cap \{ \boldsymbol p \in \mathcal{P}: \,\, p_{\boldsymbol i} \in [\delta, 1), \, \sum_{\boldsymbol i \in \mathcal{I}} p_{\boldsymbol i} = 1\},$$
where $\delta > 0$ is small enough so $\mathcal{H}_{\lambda, \delta}$ is not empty.

If $M_t$, for $t \in \{1, \dots, T\}$, is an interaction of order $h_t$, then the conditional odds ratio of $M_t$ given  the variables in $\bar{M}_t$ is the ratio of the product of some $2^{h_t}$ cell probabilities and the product of a disjoint set of $2^{h_t}$ cell probabilities.  Since $p_{\boldsymbol i} \in [\delta, 1)$ for all $\boldsymbol i \in \mathcal{I}$, the interaction  parameter $\,\,|\gamma_t(\boldsymbol p)| \leq 2^{h_t}\mbox{log } ((1-\delta)/\delta)$, and, therefore,
$$|\gamma_t(\boldsymbol p)| \leq C(\delta), \quad \mbox{for } t=1, \dots, T,$$
where $C(\delta) = 2^{{\operatorname{max}}\{h_1, \dots, h_T\}} \mbox{log } ((1-\delta)/\delta)$. Here, $C(\delta)$ is an upper bound on the interaction parameters (it plays the same role as the constant $M$ in Assumption (A4) of \cite{KalischBullm} for the Gaussian setting).

\begin{theorem} 
Let $\mathbf Y$ have a multinomial distribution with parameters $N$ and $\boldsymbol p$. Assume that, under the log-linear model corresponding to $\mathcal{H}$, the maximum likelihood estimates of the interaction parameters 
$$\hat{\boldsymbol \gamma}^{(N)}(\boldsymbol p)=(\hat{\gamma}_1^{(N)}(\boldsymbol p), \dots, \hat{\gamma}_T^{(N)}(\boldsymbol p))=({\gamma}_1^{(N)}(\hat{\boldsymbol p}), \dots, {\gamma}_T^{(N)}(\hat{\boldsymbol p}))$$
exist and are unique. Then, $\hat{\boldsymbol \gamma}^{(N)}(\boldsymbol p)$ is a uniformly, over $\mathcal{H}_{\lambda, \delta}$, consistent estimator of ${\boldsymbol \gamma}(\boldsymbol p)$.
\end{theorem}

\begin{proof}
For $t \in \{1, \dots, T\}$, $\,\gamma_t(\boldsymbol p) = \boldsymbol c_{t}' \operatorname{log } \boldsymbol p$, where $\boldsymbol c_t$ is a vector in $\mathbb{Z}^{|\mathcal{I}|}$ whose components are comprised of equal number, $2^{h_t}$, of  $1$'s and $-1$'s, and some $0$'s. By Theorem 14.6-4 in \cite{BFH}, as $N \to \infty$, $\gamma_t(\boldsymbol p)$ is asymptotically normal with mean zero and variance  
\begin{equation}\label{asVar}
{\var}(\hat{\gamma}_t) = \frac{1}{N}\boldsymbol c'_t \diag^{-1}(\boldsymbol p)\boldsymbol c_t.
\end{equation}
For every $\boldsymbol p \in \mathcal{H}_{\lambda, \delta}$,
\begin{equation}\label{VarBound}
{\var}(\hat{\gamma}_t) \leq \frac{\boldsymbol c_t'\boldsymbol c_t}{N\delta}, 
\end{equation}
and thus, using the Chebyshev inequality,
\begin{eqnarray}\label{ConsBound}
\mathbb{P}\left(|\hat{\gamma}_t^{(N)}(\boldsymbol p) - \gamma_t(\boldsymbol p)| < \epsilon\right) &=& \mathbb{P}\left(|Z|< \frac{\sqrt{N} \epsilon}{\sqrt{\boldsymbol c_{t}'\diag^{-1}(\boldsymbol p)\boldsymbol c_{t}}}\right)  
\geq \mathbb{P}\left(|Z| < \epsilon\sqrt{\frac{N\delta}{\boldsymbol c_t'\boldsymbol c_t}}\right)  \nonumber \\ 
\nonumber \\
\nonumber \\
&\geq& 1 - \frac{\boldsymbol c_t'\boldsymbol c_t}{N\delta\epsilon^2} \geq 1 - \frac{\underset{t =1, \dots, T}{\operatorname{max}}(\boldsymbol c_t'\boldsymbol c_t)}{N\delta\epsilon^2},   \quad \forall \epsilon > 0,
\end{eqnarray}
where $Z$ is a random variable with a standard normal distribution. Therefore, $\mathbb{P}(|\hat{\gamma}_t^{(N)}(\boldsymbol p) - {\gamma_t}(\boldsymbol p)| < \epsilon) \to 1$ for every $t = 1, \dots,T$, uniformly over $\boldsymbol p \in \mathcal{H}_{\lambda, \delta}$. Since the lower bound in (\ref{ConsBound}) does not depend on $t$, the proof is complete.
\end{proof}

We now address the question of how to select the threshold $\lambda$ in a hypergraph learning procedure.  We fix a $t \in \{1, \dots, T\}$ and consider testing the ``one-hyperedge'' hypothesis $H_{0t}: \, \gamma_t = 0$ versus $H_{1t}: \, \gamma_t \neq 0$ under a significance level $\alpha$. Let $\hat{\boldsymbol p}$ be the observed distribution and let $\hat{\gamma}_t = \gamma_t(\hat{\boldsymbol p}) = \boldsymbol c_t \mbox{log }\hat{\boldsymbol p}$ denote the corresponding interaction parameter. By Slutsky's Theorem, 
$$\sqrt{N} \frac{\hat{\gamma}_t}{\sqrt{\boldsymbol c'_t \diag^{-1}(\hat{\boldsymbol p})\boldsymbol c_t}} \to N(0, 1), \, \mbox{ as } N \to \infty,$$
and thus we reject the null hypothesis if 
\begin{equation}\label{TestSt}
\frac{|\hat{\gamma}_t|}{\sqrt{ \frac{1}{N}\boldsymbol c'_t \diag^{-1}(\hat{\boldsymbol p})\boldsymbol c_t}} > z_{1-\alpha/2},
\end{equation}
where $z_{1-\alpha/2} = \Phi^{-1}(1-\alpha/2)$ is the corresponding quantile of the standard normal distribution. With such a procedure, the probability of wrongly rejecting the null hypothesis does not exceed $\alpha$. 

\begin{theorem}\label{powerTh}
Let $\epsilon \in (0,1/2)$ and set 
\begin{equation}\label{LambdaVar}
\lambda^*_N = \frac{z_{1-\alpha/2}}{N^{1/2-\epsilon}}{\underset{t =1, \dots, T}{\operatorname{min}}\sqrt{\boldsymbol c_t'\boldsymbol c_t}}. 
\end{equation}
For the distributions that are $\lambda^*_N$-strong-faithful to $\mathcal{H}$, the 
 power of  the one-hyperedge test approaches $1$ as $N \to \infty$.
\end{theorem}

\begin{proof}
The asymptotic variance of $\hat{\gamma}_t$ is bounded below by  
$${\var}(\hat{\gamma}_t) = \frac{1}{N}\boldsymbol c'_t \diag^{-1}(\hat{\boldsymbol p})\boldsymbol c_t \geq \frac{1}{N} \boldsymbol c_t'\boldsymbol c_t.$$
Since for an $h_t$-order interaction, the vector $\boldsymbol c_t$  has  $2^{h_t}$ components equal to $1$,  $2^{h_t}$  components equal to $-1$, and the remaining components equal to zero, we have  $\sqrt{\boldsymbol c_t'\boldsymbol c_t} = 2^{(h_t + 1)/2}$. The distributions that are $\lambda^*_N$-strong-faithful to $\mathcal{H}$ satisfy (\ref{TestSt}) for all $t = 1, \dots, T$.
For these distributions  the power of  the one-hyperedge test  is bounded below:
$$\Phi\left(\frac{|\hat{\gamma}_t|}{\sqrt{ \frac{1}{N}\boldsymbol c'_t \diag^{-1}(\hat{\boldsymbol p})\boldsymbol c_t}} - z_{1-\alpha/2}\right)  \geq \Phi\left(z_{1-\alpha/2}(\frac{\sqrt{\boldsymbol c_t'\boldsymbol c_t}N^{-1/2+\epsilon}}{\sqrt{ \frac{1}{N}\boldsymbol c'_t \diag^{-1}(\hat{\boldsymbol p})\boldsymbol c_t}} - 1)\right) \textrm{ for all } t\in\{1,\dots T\}$$
and approaches $1$ as $N \to \infty$. 
\end{proof}

Examples of  $\lambda^*_N$ computed for hyperedges of different sizes are shown in Table \ref{lambdasTest}.  In this paper, we do not investigate any multiple comparison issues arising with testing several one-hyperedge hypotheses at the same time.

\begin{table}
\centering
\caption{Possible threshold values for the parameter $\lambda$.}
\label{lambdasTest}
\vspace{3mm}

\begin{tabular}{l|c|l}
\hline
&  & \\
Hyperedge & \multicolumn{1}{c|}{The order of the odds ratio} & \multicolumn{1}{c}{$\lambda^*_N$} \\
&  &  \\
\hline
 & & \\
$[AB]$ & $h = 1$ & $\lambda^*_N = \frac{z_{1-\alpha/2}}{N^{1/4}} \cdot 2$\\
& & \\
$[ABC]$ & $h=2$ & $\lambda^*_N = \frac{z_{1-\alpha/2}}{N^{1/4}} \cdot  2^{3/2}$ \\
& & \\
$[ABCD]$ & $h=3$ & $\lambda^*_N = \frac{z_{1-\alpha/2}}{N^{1/4}}\cdot 4$\\
& & \\
$[ABCDE]$ & $h=4$ & $\lambda^*_N = \frac{z_{1-\alpha/2}}{N^{1/4}}\cdot  2^{5/2}$\\
\hline
\end{tabular}
\end{table}

For learning a hypergraph, any model selection procedure for hierarchical log-linear models can be applied. A review of such procedures can be found in \cite{EdwardsBook}. Backward selection which starts from the saturated model and, using the edge removal mechanism described by \cite {EdwardsNote}, goes through a sequence of nested hypergraphs, is a polynomial time algorithm that is appropriate for high dimensions. Uniform consistency of the maximum likelihood estimates for the maximal interactions of the distributions in $\mathcal{H}_{\lambda, \delta}$ entails that backward selection is a uniformly consistent procedure and the hypergraph will be determined correctly.

\section{Proportions of strong-unfaithful distributions}\label{SecProp}

As shown in the previous section, strong-faithfulness ensures the existence of uniformly consistent tests for developing methods for learning the underlying hypergraph. If the parameter space has finite volume it is possible to estimate the proportion (relative volume) of the distributions that are not $\lambda$-strong-faithful to a  model of interest and thus whose association structure may be discovered incorrectly. \cite{UhlerFaithGeometry} analyzed the proportion of distributions that are not $\lambda$-strong-faithful to a DAG in the Gaussian setting. Partial correlations define varieties and strong-unfaithful distributions correspond to the parameters that lie in a tube around these varieties. So the relative volume of these tubes corresponds to the proportion of distributions that don't satisfy the strong-faithfulness assumption, and lower bounds on these volumes were given for different classes of DAGs. The following example illustrates how one can estimate such volumes in the discrete case.

\vspace{3mm}

\textbf{Example \ref{22param}} (revisited):
Consider a hierarchical log-linear  parameterization of the distributions on the $2\times 2$ contingency table:
\begin{align}\label{22corner}
\operatorname{log} p_{00} &=\gamma^{\emptyset}, \nonumber \\ 
\operatorname{log} p_{01} &=\gamma^{\emptyset}+ \gamma^{B}_{1},\\
\operatorname{log} p_{10} &= \gamma^{\emptyset} + \gamma^{A}_{1}, \nonumber \\ 
\operatorname{log} p_{11} &= \gamma^{\emptyset} + \gamma^{A}_{1}+\gamma^{B}_{1}+ \gamma^{AB}_{11}. \nonumber
\end{align}
The interaction parameter, $\gamma_{11}^{AB}$, which was denoted by $\phi_1$ in Example \ref{22param} and in the corresponding Table \ref{StrFtwoway} and Figure \ref{2by2Prop}, can be expressed in terms of conditional probabilities $\theta_1 = \mathbb{P}(A = 0)$, $\theta_2 = \mathbb{P}(B=0\mid A=0)$, and $\theta_3 = \mathbb{P}(B=0\mid A=1)$:
\begin{equation*}
\gamma^{AB}_{11} =\mbox{log } \left(\frac{p_{00}p_{11}}{p_{01}p_{10}}\right)= \mbox{log } \frac{\theta_2(1-\theta_3)}{(1-\theta_2)\theta_3} = \mbox{log } \frac{\theta_2}{1-\theta_2} - \mbox{log } \frac{\theta_3}{1-\theta_3}.
\end{equation*}
Let 
\begin{eqnarray*}
\mathcal{H}_{\lambda} &=&\left \{(\theta_1, \theta_2, \theta_3) \in (0,1)^3: \,\, \left|\mbox{log } \frac{\theta_2}{1-\theta_2} - \mbox{log } \frac{\theta_3}{1-\theta_3}\right| > \lambda \right\}.
\end{eqnarray*}
The volume of its complement, $\bar{\mathcal{H}}_{\lambda}$, is equal to:
\begin{eqnarray}
\vol(\bar{\mathcal{H}}_{\lambda})  &=& \vol \left \{(\theta_1, \theta_2, \theta_3) \in (0,1)^3: \,\, e^{-\lambda} < \frac{\theta_2}{1-\theta_2} \cdot \frac{1-\theta_3}{\theta_3}  < e^{\lambda} \right \}\nonumber\\ 
&=& \int_{0}^1 d \theta_2 \left(  
\frac{\theta_2}{\theta_2(1-e^{-\lambda}) + e^{-\lambda}} - 
\frac{\theta_2}{\theta_2(1-e^{\lambda}) + e^{\lambda}}
\right) \nonumber\\
{}\nonumber\\
&=& \frac{e^{2\lambda} - 2\lambda e^{\lambda} - 1}{(1-e^{\lambda})^2},\label{eq_gamma2}
\end{eqnarray}
where the integral was computed by substitution. The parameter space $(0,1)^3$ has a unit volume. Hence, the relative proportion of distributions that are not $\lambda$-strong-faithful to $[AB]$ is equal to $\vol(\bar{\mathcal{H}}_{\lambda}) $. For small $\lambda$ this proportion is approximately $\frac{\lambda}{3}$, which is consistent with the simulation results for $\phi_1$ in Figure \ref{2by2Prop}.\qed

\vspace{3mm}

\begin{theorem}\label{BoundMax}
Let $\mathcal{H}$ be a hypergraph whose hyperedges $M_1, \dots, M_T$ are interactions of order $h_1, \dots, h_T$ respectively, and let 
$\bar{\mathcal{H}}_{\lambda}$ be the set of distributions that are not $\lambda$-strong-faithful to $\mathcal{H}$.  Then,
$$\vol(\bar{\mathcal{H}}_{\lambda}) \geq \underset{t \in \{1, \dots, T\}}{\operatorname{max}} \vol \{\boldsymbol p \in \mathcal{H}: \,\, |\gamma_{t}(\boldsymbol p)| < \lambda\} \geq \underset{t \in \{1, \dots, T\}}{\operatorname{max}}\left(\frac{e^{2\mu} - 2\mu e^{\mu} - 1}{(1-e^{\mu})^2}\right)^{2^{h_t-1}},$$
where $\mu=\lambda/2^{h_t-1}$ 
\end{theorem}

\begin{proof}

By Lemma \ref{interactions}, the parameter $\gamma_{t}$ for $t \in \{1, \dots, T\}$ is the log conditional odds ratio of $M_t$, given the variables in $\bar{M}_t$. We can express $\gamma_t$ using a corresponding set of variation independent conditional probabilities:
\begin{eqnarray*}
\gamma_{t} = \mbox{log } \left(\frac{\theta_{1}}{1-\theta_{1}} \cdots \frac{\theta_{2^{h_t-1}}}
{1-\theta_{2^{h_t-1}}}\cdot
\frac{1-\theta_{2^{h_t-1}+1}}{\theta_{2^{h_t-1}+1}} \cdots \frac{1-\theta_{2^{h_t}}}{ \theta_{2^{h_t}}}\right).\end{eqnarray*}
Therefore, 
\begin{eqnarray*}
&&\vol \{\boldsymbol p \in \mathcal{H}: \,\, |\gamma_{t}(\boldsymbol p)| < \lambda\} = \vol \left\{\boldsymbol \theta \in (0,1)^m: \,\, |\gamma_{t}| < \lambda \right\}\\
&& \\
&=&\vol\left\{(\theta_1, \dots, \theta_{2^{h_t}})\!\in\!(0,1)^{2^{h_t}}\!:\, \left| \mbox{log }\!\!\left(\!\frac{\theta_{1}}{1-\theta_{1}} \cdots \frac{\theta_{2^{h_t-1}}}
{1-\theta_{2^{h_t-1}}}\cdot
\frac{1-\theta_{2^{h_t-1}+1}}{\theta_{2^{h_t-1}+1}} \cdots \frac{1-\theta_{2^{h_t}}}{ \theta_{2^{h_t}}}\!\right)\!\right| < \lambda \right\} \\ 
&\geq&
\left(\vol\,\left\{(\zeta_1, \zeta_2)\in(0,1)^2:\,\, \left| \mbox{log } \frac{\zeta_1}{1-\zeta_1} - \mbox{log } \frac{\zeta_2}{1-\zeta_2}\right| < \frac{\lambda}{2^{h_t-1}}\right\} \right)^{2^{h_t-1}} \\
&=& \left(\frac{e^{2\mu} - 2\mu e^{\mu} - 1}{(1-e^{\mu})^2}\right)^{2^{h_t-1}},
\end{eqnarray*}

where $\mu = {\lambda}/{2^{h_t-1}}$, and for the last equation we used (\ref{eq_gamma2}). Since $\bar{\mathcal{H}}_{\lambda} = \{\boldsymbol p \in \mathcal{H}: \,\, |\gamma_t(\boldsymbol p)| < \lambda \mbox{ for at least one } t \in \{1, \dots, T\}\}$,
$$\vol(\bar{\mathcal{H}}_{\lambda}) \geq \vol \{\boldsymbol p \in \mathcal{H}: \,\, |\gamma_t(\boldsymbol p)| < \lambda \mbox{ for all } t \in \{1, \dots, T\}\},$$
and thus
$$\vol(\bar{\mathcal{H}}_{\lambda}) \geq \underset{t \in \{1, \dots, T\}}{\operatorname{max}} \left(\frac{e^{2\mu} - 2\mu e^{\mu} - 1}{(1-e^{\mu})^2}\right)^{2^{h_t-1}}, \quad \mbox{for } \mu = {\lambda}/{2^{h_t-1}}.$$
\end{proof}

\vspace{5mm}

As we will see in the following result, for hypergraphs whose hyperedges are variation independent we can in fact give an exact formulation of the proportion of distributions that don't satisfy strong-faithfulness.

\begin{theorem} \label{conjecture}
Let $\mathcal{H}$ be the hypergraph generated by marginals $M_1, \dots, M_T$. Assume that there exists a parameterization of $\mathcal{H}$ under which the interaction parameters corresponding to $M_1, \dots, M_T$ are variation independent and, further, that the parameter space has finite volume. Then, the proportion of distributions that are not $\lambda$-strong-faithful to $\mathcal{H}$ is
\begin{equation}\label{VolumeDecomp}
\vol(\bar{\mathcal{H}}_{\lambda}) = 1 - (1-\boldsymbol \nu_{M_1}) \cdots (1-\boldsymbol \nu_{M_T}),
\end{equation}
where $\boldsymbol \nu_{M_t}$, for $t = 1, \dots, T$, is the proportion of distributions that are not $\lambda$-strong-faithful to $M_t$.
\end{theorem}

\begin{proof}
Consider a parameterization of $\mathcal{H}$ under which the maximal interaction parameters corresponding to $M_1, \dots, M_T$ are variation independent. For $t \in \{1, \dots, T\}$, let $\boldsymbol \nu_{M_t}$ denote the proportion of distributions that do not satisfy $\lambda$-strong-faithfulness with respect to $M_t$. Since the joint range of the variation independent parameters is equal to the Cartesian product of the individual ranges,  the proportion of distributions that are not $\lambda$-strong-faithful to $\mathcal{H}$ is equal to
$1 - (1-\boldsymbol \nu_{M_1}) \cdots (1-\boldsymbol \nu_{M_T}).$
\end{proof}

The maximal interaction parameters in decomposable log-linear models  \citep{Haberman} and in ordered decomposable marginal log-linear models \citep{RudasBergsma} are variation independent, and Theorem \ref{conjecture} applies. Let $\mathcal{H}$ be the hypergraph generated by a decomposable sequence of marginals $M_1, \dots, M_T$.  The interaction parameter associated with a hyperedge $M_t$ can be expressed using a variation independent set of conditional probabilities. Under this parameterization, the proportion of distributions that do not satisfy $\lambda$-strong-faithfulness with respect to the hyperedge $M_t$ is equal to 
\begin{eqnarray}\label{nu}
\boldsymbol \nu_{h_t} = \vol\!\left\{\!(\theta_1, \dots, \theta_{2^{h_t}})\!:\, \left|\log \!\left(\frac{\theta_{1}}{1-\theta_{1}} \cdots \frac{\theta_{2^{h_t-1}}}
{1-\theta_{2^{h_t-1}}}\cdot
\frac{1-\theta_{2^{h_t-1}+1}}{\theta_{2^{h_t-1}+1}} \cdots \frac{1-\theta_{2^{h_t}}}{ \theta_{2^{h_t}}}\right)\!\right| < \lambda \right\},
\end{eqnarray}
where $h_t$ denotes the order of interaction $M_t$. The proportion of distributions that are not $\lambda$-strong-faithful to $\mathcal{H}$ is calculated using (\ref{VolumeDecomp}).

Figure \ref{Chains} shows values of $\boldsymbol \nu_{h}$ as functions of  $h$ and $\lambda$. The concrete computations involved in Theorem \ref{conjecture} are illustrated in Example \ref{ChainExample}. We next analyze hypergraphs with a special ``chain'' structure and show that in this case Equation (\ref{VolumeDecomp}) simplifies.

\begin{definition}
A hypergraph $\mathcal{H}$ is called a chain of order $h$ if the generating sequence of marginals 
$\{M_1, \dots M_T\}$, where $\cup_{t = 1}^T M_t = V$, is decomposable and all of the hyperedges correspond to $h$-th order interactions of the joint distribution. 
\end{definition}

For example, the hypergraph generated by $\{A,B\}$,  $\{B,C\}$, $\{C,D\}$, $\{D,E\}$ is a chain of order $1$ of length $4$, and the hypergraph generated by $\{A,B,C\}$ and $\{A,B,D\}$ is a chain of order 2 of length 2.

\begin{corollary} \label{ChainH}
Let a hypergraph $\mathcal{H}$ be a chain of order $h$ of length $L$. Then the proportion of distributions that are not $\lambda$-strong-faithful to $\mathcal{H}$ is equal to
$$1 - (1-\boldsymbol \nu_h)^L,$$
where 
\begin{eqnarray}\label{nuH}
\boldsymbol \nu_{h} = \vol\left\{(\theta_1, \dots, \theta_{2^{h}}):\,\, \left|\log \left(\frac{\theta_{1}}{1-\theta_{1}} \cdots \frac{\theta_{2^{h-1}}}
{1-\theta_{2^{h-1}}}\cdot
\frac{1-\theta_{2^{h-1}+1}}{\theta_{2^{h-1}+1}} \cdots \frac{1-\theta_{2^{h}}}{ \theta_{2^{h}}}\right)\right| < \lambda \right\}.
\end{eqnarray}
\end{corollary}

For a chain of order $1$, the proportion of distributions that are not $\lambda$-strong-faithful to $\mathcal{H}$ is especially simple:
$$\vol(\bar{\mathcal{H}}_{\lambda}) = 1 - (1-\boldsymbol \nu_1)^T,$$
where 
\begin{eqnarray}\label{nu1}
\boldsymbol \nu_1 = \frac{e^{2\lambda} - 2\lambda e^{\lambda} - 1}{(1-e^{\lambda})^2}.
\end{eqnarray}

The proportions for chains of several orders were estimated using Monte-Carlo method and are displayed in Figure \ref{ChainsProportions}.

\begin{figure}[!t]
\begin{minipage}[t]{0.43\textwidth}
\vspace{0pt}
\centering
\includegraphics[width=\linewidth]{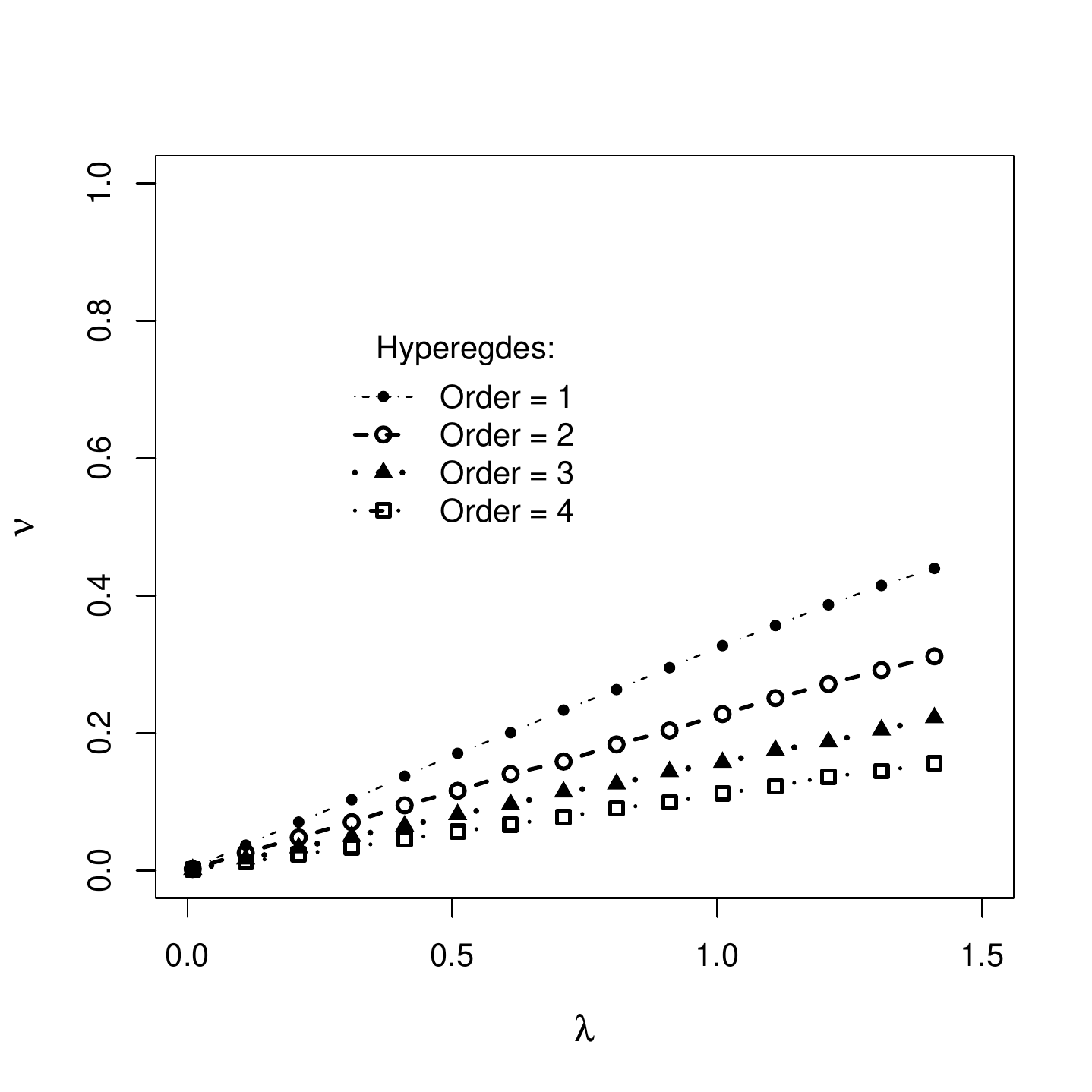}
\caption{Proportions of distributions that do not satisfy strong-faithfulness with respect to a single hyperedge. See Equation (\ref{nu}). }
\label{Chains}
\end{minipage}
\hfill
\begin{minipage}[t]{0.43\textwidth}
 \vspace{0pt}
\centering
\includegraphics[width=\linewidth]{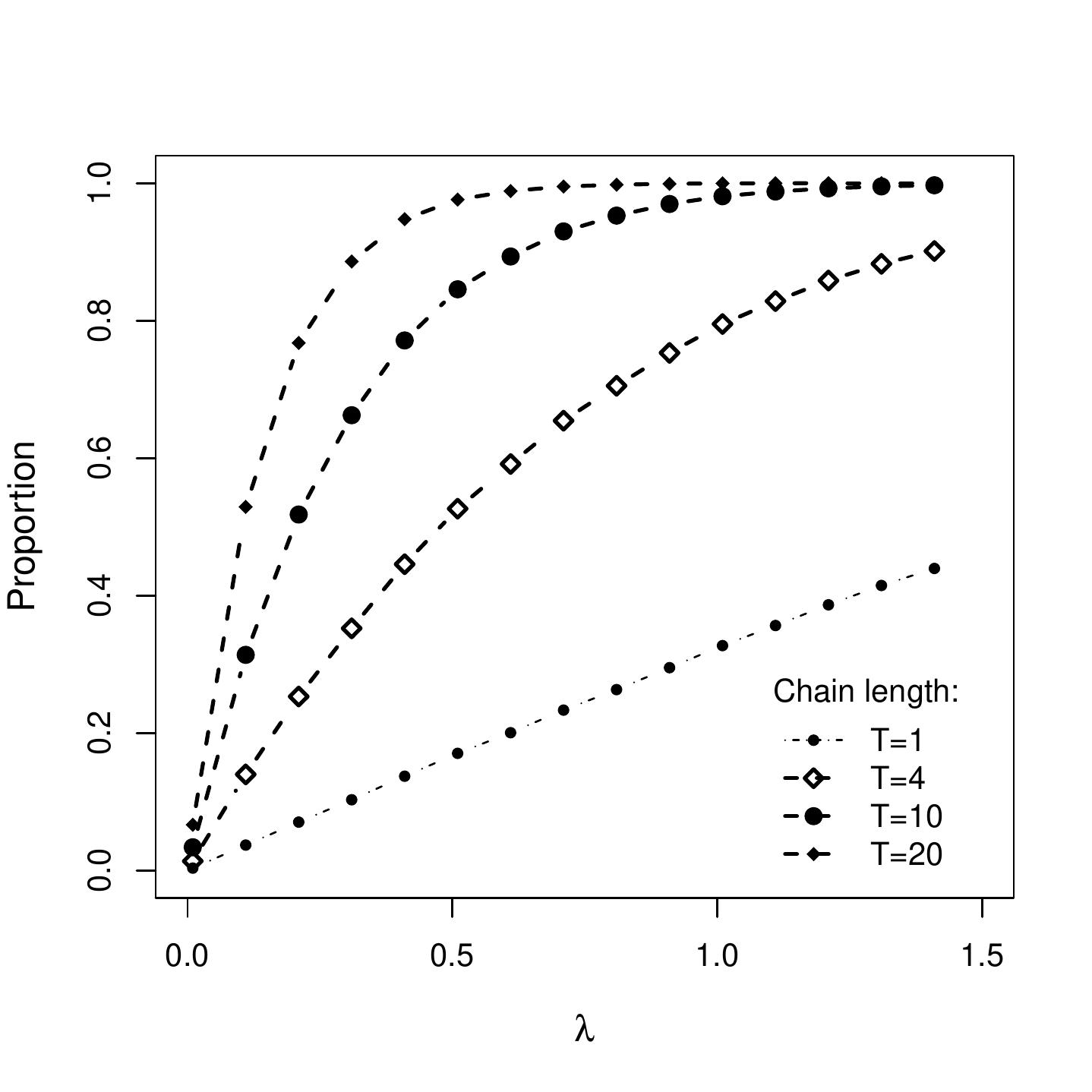}
\caption{Proportions of distributions that are not $\lambda$-strong-faithful to a first order chain. See Corollary \ref{ChainH}.}
\label{ChainsProportions}
\end{minipage}
\end{figure}

\begin{example} \label{ChainExample}

We demonstrate the volume computation using the chain $[AB][BC][CD]$ of order 1.
The maximal interaction parameters corresponding to the hyperedges are: 
\begin{eqnarray*}
\gamma_1 &=&  \mbox{log } \mathcal{COR}(AB \mid CD) = \mbox{log } \frac{p_{00kl}p_{11kl}}{p_{01kl}p_{10kl}}, \\
\gamma_2 &=& \mbox{log } \mathcal{COR}(BC \mid AD) = \mbox{log } \frac{p_{i00l}p_{i11l}}{p_{i01l}p_{i10l}}, \\ 
\gamma_3 &=&  \mbox{log } \mathcal{COR}(CD \mid AB) = \mbox{log } \frac{p_{ij00}p_{ij11}}{p_{ij01}p_{ij10}},
\end{eqnarray*}
where $i, j, k, l \in \{0, 1\}$ are fixed categories of $A$, $B$, $C$, $D$ respectively. The chain can be described by two conditional independence relations:  $A \ci C \mid B$ and $AB \ci D \mid C$. Thus the distributions in a chain model can be parameterized using the conditional probabilities:
\begin{align*}
&\theta_{0} = \mathbb{P}(B = 0), \\
&\theta_{10} = \mathbb{P}(A=0 \mid B = 0), \quad \theta_{11} = \mathbb{P}(A = 0 \mid B = 1),  \\
&\theta_{20} = \mathbb{P}(C = 0 \mid B = 0), \quad \theta_{21} = \mathbb{P}(C = 0 \mid B = 1),  \\
&\theta_{30} = \mathbb{P}(D = 0 \mid C = 0), \quad \theta_{31} = \mathbb{P}(D = 0 \mid C = 1).
\end{align*}
The parameters $\boldsymbol \theta$ are variation independent, and, for $t = 1, 2, 3$,
$$\gamma_t = \mbox{log } \left(\frac{\theta_{t0}}{1-\theta_{t0}} \cdot \frac{1-\theta_{t1}}{\theta_{t1}}\right).$$
Let 
\begin{eqnarray*}
\boldsymbol \nu_1  &=& \vol\left \{(\theta_1, \theta_2) \in (0,1)^2: \,\,  \left|\mbox{log } \frac{\theta_1}{1-\theta_1}  - \mbox{log } \frac{\theta_2}{1-\theta_2} \right| < \lambda \right\} =\frac{e^{2\lambda} - 2\lambda e^{\lambda} - 1}{(1-e^{\lambda})^2}.
\end{eqnarray*} 
Using the binomial formula, we obtain that
$\vol(\bar{\mathcal{H}}_{\lambda}) = 1 - (1 - \boldsymbol \nu_1)^3$. \qed
\end{example}


\vspace{5mm}

\textbf{Example \ref{FourVarUnfaith}} (revisited):
The maximal non-vanishing interactions of a distribution that is faithful to a hypergraph $\mathcal{H}$ with hyperedges $\{A,B,C\}$ and $\{A,B,D\}$ can be described using the interaction parameters equal to the logarithm of the second order conditional odds ratios of $ABC$ given $D$ and of $ABD$, given $C$: 
\begin{eqnarray*}
\gamma^{ABC}_0 &=& \mbox{log } \mathcal{COR}(ABC\mid D=0),\\
\gamma^{ABC}_1 &=& \mbox{log } \mathcal{COR}(ABC\mid D=1),\\
\gamma^{ABD}_0 &=& \mbox{log } \mathcal{COR}(ABD\mid C=0),\\
\gamma^{ABD}_1 &=& \mbox{log } \mathcal{COR}(ABD\mid C=1).
\end{eqnarray*}
Using conditional probabilities,
\begin{eqnarray*}
\theta_1 &=& \mathbb{P}(C=0\mid A=0, B=0), \quad \theta_2 = \mathbb{P}(C=0\mid A=0, B=1), \\
\theta_3 &=& \mathbb{P}(C=0\mid A=1, B=0), \quad \theta_4 = \mathbb{P}(C=0\mid A=1, B=1), \\
\theta_5 &=& \mathbb{P}(D=0\mid A=0, B=0), \quad \theta_6 = \mathbb{P}(D=0\mid A=0, B=1), \\
\theta_7 &=& \mathbb{P}(D=0\mid A=1, B=0), \quad \theta_8 = \mathbb{P}(D=0\mid A=1, B=1),
\end{eqnarray*} 
one obtains
\begin{eqnarray*}
\gamma^{ABC}_0 &=& \gamma^{ABC}_1 = \mbox{log } \frac{\theta_1}{1-\theta_1} + \mbox{log } \frac{\theta_4}{1-\theta_4} - \mbox{log } \frac{\theta_2}{1-\theta_2} - \mbox{log } \frac{\theta_3}{1-\theta_3}, \\
\gamma^{ABD}_0 &=& \gamma^{ABD}_1 = \mbox{log } \frac{\theta_5}{1-\theta_5} + \mbox{log } \frac{\theta_8}{1-\theta_8} - \mbox{log } \frac{\theta_6}{1-\theta_6} - \mbox{log } \frac{\theta_7}{1-\theta_7}.
\end{eqnarray*}

A distribution is not $\lambda$-strong-faithful to the hypergraph $\mathcal{H}$ if at least one of the following inequalities holds:
$$|\gamma^{ABC}_0| < \lambda, \mbox{ or } |\gamma^{ABD}_0| < \lambda.$$
Hence, the proportion of distributions that are not $\lambda$-strong-faithful to the hypergraph $\mathcal{H}$ is equal to $1 - (1 - \boldsymbol \nu_2)^2$,
where \begin{eqnarray*}
\boldsymbol \nu_2  &=& \vol\left \{(\theta_1, \dots, \theta_4) \in (0,1)^4: \, \left|\mbox{log } \frac{\theta_1}{1-\theta_1} + \mbox{log } \frac{\theta_4}{1-\theta_4} - \mbox{log } \frac{\theta_2}{1-\theta_2} - \mbox{log } \frac{\theta_3}{1-\theta_3} \right| < \lambda \right\}.
\end{eqnarray*} 

We were not able to find a closed-form expression for $\boldsymbol \nu_2$. It can be shown that $\boldsymbol \nu_2$ is bounded above by $\boldsymbol \nu_1$ and thus the volume of distributions that are not $\lambda$-strong-faithful to the hypergraph $\mathcal{H}$ is bounded above by the volume computed for the chain of the same length of order $1$, that is,
$\vol(\bar{\mathcal{H}}_{\lambda}) \leq 1 - (1-\boldsymbol \nu_{1})^2.$
 \qed

\begin{remark}
The concept of strong-faithfulness can be extended to distributions that do not belong to a given hypergraph model. Let $\boldsymbol p, \boldsymbol q \in \mathcal{P}$, and let $\rho$ be a divergence function. The distance from a distribution $\boldsymbol p$ to a hypergraph $\mathcal{H}$ can be defined as 
\begin{equation}\label{distModel}
\rho(\boldsymbol p, \mathcal{H}) = \underset{\boldsymbol q \in \mathcal{H}}{\mbox{min }} \rho(\boldsymbol p, \boldsymbol q).
\end{equation}
In particular, $\rho(\boldsymbol p, \mathcal{H}) = 0$ if and only if $\boldsymbol p \in \mathcal{H}$. We denote by $\boldsymbol p_{\mathcal{H,\rho}}$ the projection of $\boldsymbol p$ onto the hypergraph $\mathcal{H}$, i.e.:
$$\boldsymbol p_{\mathcal{H}, \rho} = \underset{\boldsymbol q \in \mathcal{H}}{\mbox{argmin }} \rho(\boldsymbol p, \boldsymbol q),$$
and call a distribution $\boldsymbol p$ \emph{projected-$\lambda$-strong-faithful to} $\mathcal{H}$ with respect to $\rho$ (for $\lambda >0$) if  $\boldsymbol p_{\mathcal H, \rho}$ is $\lambda$-strong-faithful to $\mathcal{H}$. The concept of projected-strong-faithfulness is relevant in various estimation procedures.
\end{remark}

We end by illustrating the concept of projected-strong-faithfulness by estimating the proportions of projected-$\lambda$-strong-faithful distributions for several hypergraph models on the $2 \times 2 \times 2$ contingency table.

\begin{figure}[!t]
\centering
\includegraphics[scale=0.45]{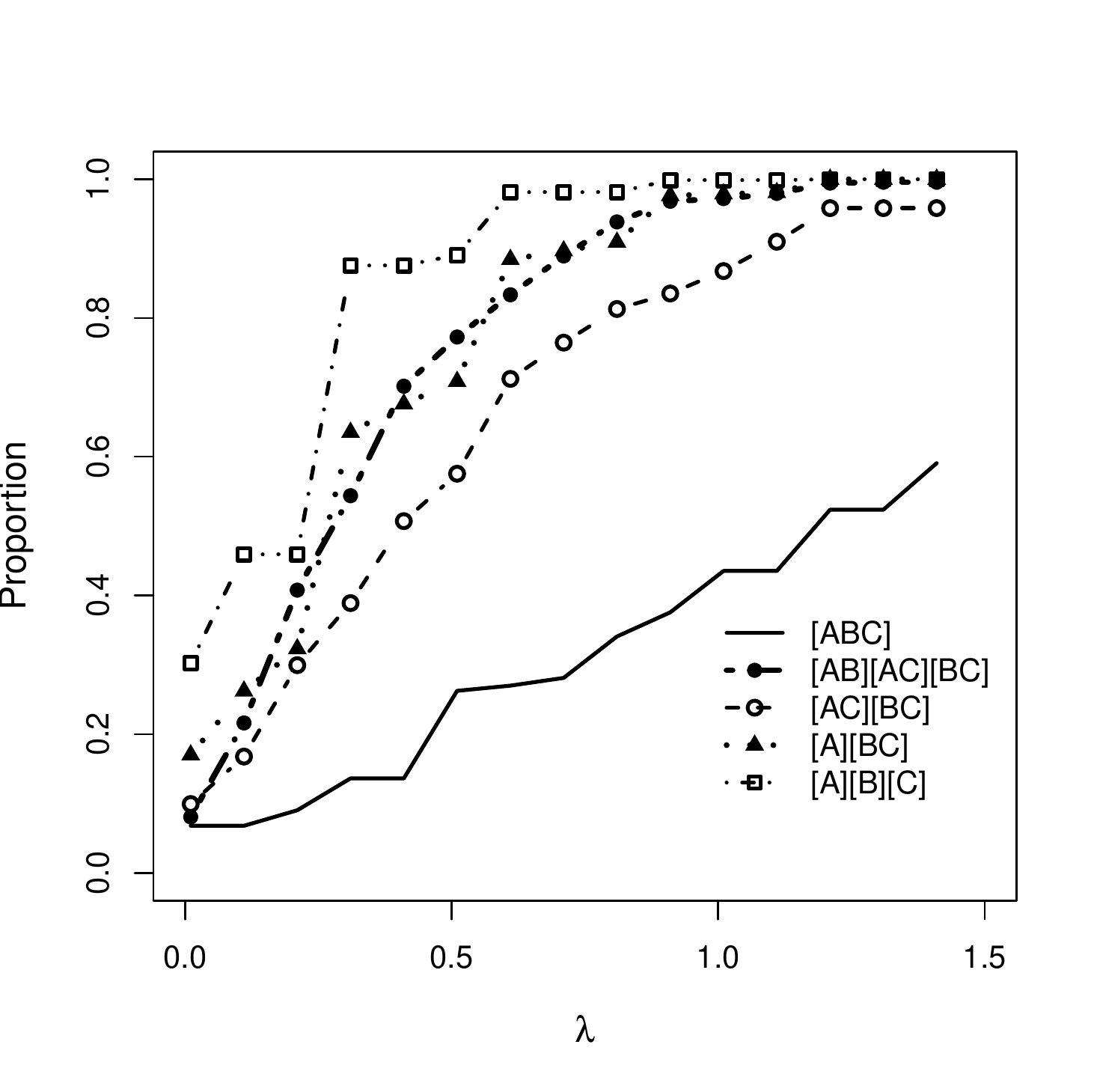}
\vspace{-0.5cm}\caption{Proportions of distributions that are not projected-$\lambda$-strong-faithful computed for several hypergraphs on the $2 \times 2 \times 2$ contingency table.}
\label{C1}
\end{figure}

\vspace{4mm}

\textbf{Example \ref{222marginal}} (revisited): 
To determine the distance from a distribution $\boldsymbol p$ to a hypergraph $\mathcal{H}$ we use the likelihood function under the corresponding log-linear model. Relative frequencies of distributions that do not satisfy the projected-$\lambda$-strong-faithfulness relation for different hypergraphs and different values of $\lambda$ are displayed in Figure \ref{C1}.  
\qed

\section{Conclusion}\label{SecConcl}

We demonstrated that the association structure of discrete data can be very complex, and some distributions are not faithful to any undirected graphical model or any DAG. Thus, the attractive simplicity of graphical models may be misleading. In Section \ref{sectionGraphFaith}, we proposed the concept of parametric faithfulness, which can be applied to any exponential family, including those  which cannot be specified using Markov properties. 

We considered the class of hypergraphs which can be identified with hierarchical log-linear models. We showed that for any distribution there exists a hypergraph to which it is parametrically faithful and suggested to conduct the search in this class. As the class also  contains graphical models, if a model structure which can be described by a graph is appropriate, it will be discovered (see the consistency result in Section \ref{sectionStrongFaith}).  

Our work is relevant for the popular causal search algorithms, referred to in Section \ref{intro}, which assume (strong-) faithfulness.  The findings described in Sections \ref{sectionStrongFaith} and \ref{SecProp} imply that, depending on the quantitative expression for association and on the choice of the cut-off parameter $\lambda$, to define strong-faithfulness, the resulting model selection procedures may yield different results for the same data.

\section*{Acknowledgements}

The authors wish to thank Antonio Forcina and Ren{\'a}ta N{\'e}meth for helpful comments and discussions. The third author was supported in part by Grant K-106154 from the Hungarian National Scientific Research Fund (OTKA).

\bibliographystyle{plainnat}
\bibliography{uwthesis0507faith}

\begin{table}
\centering
\label{allgraphsABC}
\caption{Some graphical models on three nodes.}
\vspace{6mm}

\begin{tabular}{m{7mm}m{32mm}m{65mm}m{35mm}}
\hline 
&  &  & \\
&  \multicolumn{1}{c}{Log-Linear Model} &  \multicolumn{1}{c}{Conditional Independence} &  \multicolumn{1}{c}{Graph}  \\ [2ex]
\hline
1 & \multicolumn{1}{c}{[ABC]} &  \multicolumn{1}{c}{None} &  \multicolumn{1}{c}{\begin{minipage}{.2\textwidth}\includegraphics[scale=0.5]{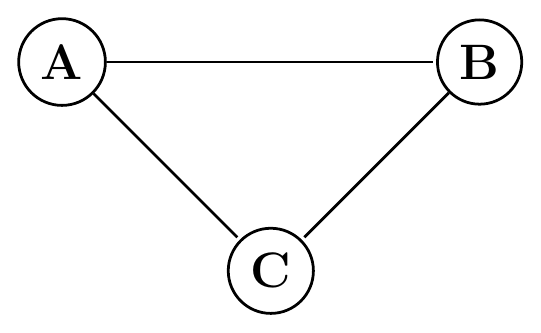}\end{minipage}}
 \\ [4ex]
\hline 
2 &  \multicolumn{1}{c}{[A][B]} &  \multicolumn{1}{c}{$A \ci B$} &  \multicolumn{1}{c}{\begin{minipage}{.2\textwidth}\includegraphics[scale=0.5]{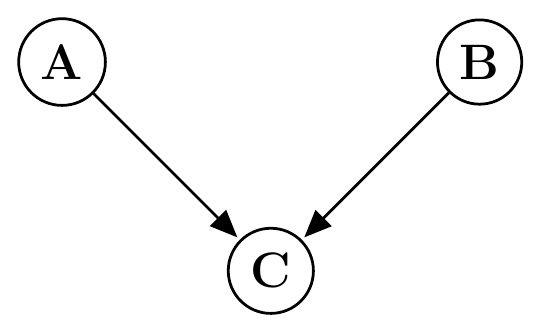}\end{minipage}}  \\ [4ex]
\hline
3  &  \multicolumn{1}{c}{[A][C]} &   \multicolumn{1}{c}{$A \ci C$} &  \multicolumn{1}{c}{\begin{minipage}{.2\textwidth}\includegraphics[scale=0.5]{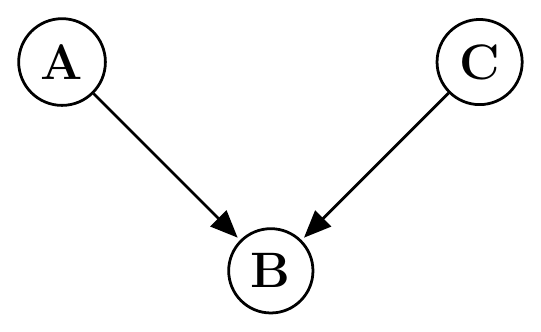}\end{minipage}} \\ [4ex]
\hline
4 &  \multicolumn{1}{c}{[B][C]} &  \multicolumn{1}{c}{$B \ci C$} &  \multicolumn{1}{c}{\begin{minipage}{.2\textwidth}\includegraphics[scale=0.5]{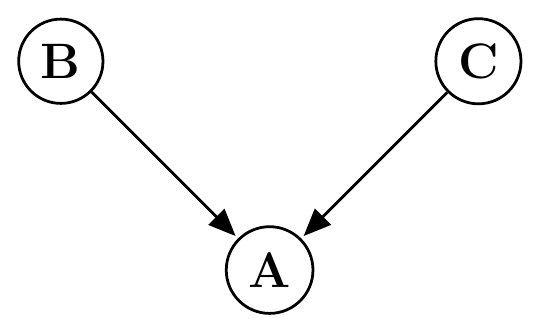}\end{minipage}} \\ [4ex]
\hline
5 &  \multicolumn{1}{c}{[AB][C]} &  \multicolumn{1}{c}{$AB \ci C$} &  \multicolumn{1}{c}{\begin{minipage}{.2\textwidth}\includegraphics[scale=0.5]{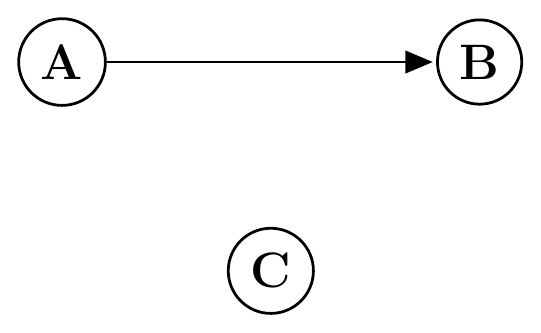}\end{minipage}}  \\ [4ex]
\hline
6 &  \multicolumn{1}{c}{[AC][B]} &  \multicolumn{1}{c}{$AC \ci B$} &  \multicolumn{1}{c}{\begin{minipage}{.2\textwidth}\includegraphics[scale=0.5]{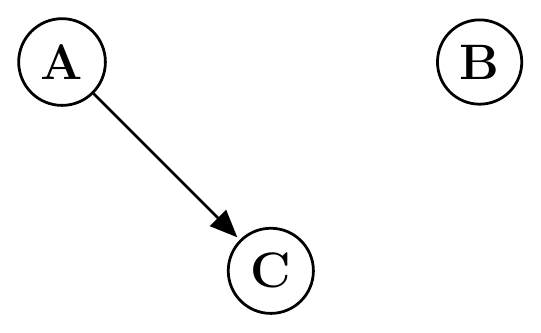}\end{minipage}}  \\ [4ex]
\hline
7 &  \multicolumn{1}{c}{[A][BC]} &  \multicolumn{1}{c}{$A \ci BC$}  &  \multicolumn{1}{c}{\begin{minipage}{.2\textwidth}\includegraphics[scale=0.5]{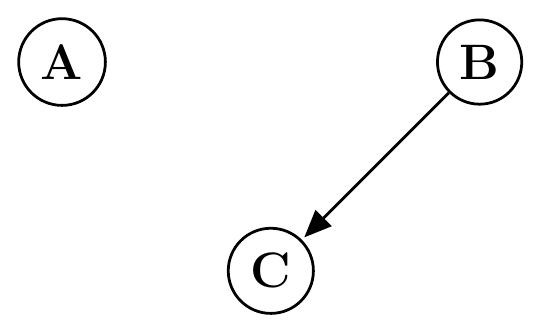}\end{minipage}} \\ [4ex]
\hline
8 &  \multicolumn{1}{c}{[A][B][C]} &  \multicolumn{1}{c}{\begin{tabular}{lll} $A \ci B$, & $A \ci C$, & $B \ci C$, \\
 $A \ci B | C$, & $A \ci C | B$, & $B \ci C | A$ \end{tabular}} &  \multicolumn{1}{c}{\begin{minipage}{.2\textwidth}\includegraphics[scale=0.5]{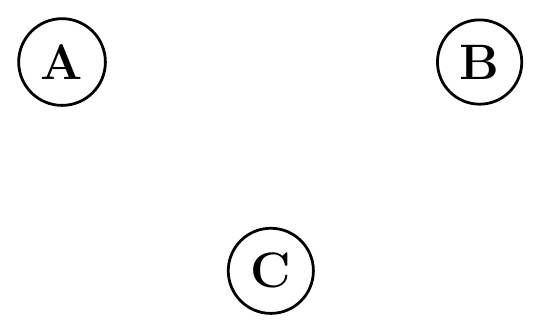}\end{minipage}} \\ [4ex]
\hline
\end{tabular}
\end{table}



\end{document}